\documentclass[11pt]{article}
\usepackage{geometry}
\usepackage{amsthm}
\usepackage{amsmath}
\usepackage{amsfonts}
\usepackage{graphicx}
\newtheorem{thm}{Theorem}[section]

\newtheorem{lem}[thm]{Lemma}

\theoremstyle{definition}
\newtheorem{defn}[thm]{Definition}

\theoremstyle{remark}

\numberwithin{equation}{section}
\newcommand{\eps}{\varepsilon}

\def\id{\operatorname{Id}}

\def\sign{\operatorname{sign}}
\def\CT{{\cal C}}
\def\R{\mathbb R}
\def\eps{\varepsilon}

\def\errmat{{\mathcal E}}
\renewcommand{\cos}{\operatorname{cos}}
\renewcommand{\sin}{\operatorname{sin}}
\title{An Omega((n log n)/R) Lower Bound for Fourier Transform Computation in the R-Well Conditioned Model  }
\author{Nir Ailon \\ Department of Computer Science \\ Technion Israel Institute of Technology \\ Haifa, Israel \\ \texttt{nailon@cs.technion.ac.il}}

\begin{document}

\maketitle

% ----------------------------------------------------------------
%\begin{abstract}
%\end{abstract}
\def\sign{\operatorname{sgn}}
\def\dim{n}
\def\C{\mathbb C}
\def\R{\mathbb R}
\def\P{\mathcal P}
\def\trace{\operatorname{tr}}
\def\diag{\operatorname{diag}}
\def\rank{\operatorname{rank}}
\def\F{{\mathcal F}}
\def\Id{\operatorname{Id}}
\def\f{{\hat f}}
% ----------------------------------------------------------------

\begin{abstract}
Obtaining a non-trivial (super-linear) lower bound for computation of the Fourier transform in the linear circuit model has been a long standing open problem for over 40 years.

An early result by Morgenstern from 1973, provides an $\Omega(n \log n)$ lower bound for the unnormalized Fourier transform when the constants used in the computation are bounded.  The proof uses a  potential function related to a determinant.  That result does not explain why the normalized Fourier transform (of unit determinant) should be difficult to compute in the same model. Hence, it is not scale insensitive.

More recently, Ailon (2013) showed that if only unitary 2-by-2 gates are used, and additionally no extra memory is allowed, then the normalized Fourier transform requires $\Omega(n\log n)$ steps.  This rather limited result is also sensitive to scaling, but highlights the complexity inherent in the Fourier transform arising from introducing entropy, unlike, say, the identity matrix (which is as complex as the Fourier transform using Morgenstern's arguments, under proper scaling).

In this work we extend the arguments of Ailon (2013).  In the first extension, which is also the main contribution, we provide a lower bound for computing any scaling of the Fourier transform.  Our restriction is that,  the composition of all gates up to any point  must be a well conditioned linear transformation.    The lower bound is $\Omega(R^{-1}n\log n)$, where $R$ is the uniform condition number.   The model is realistic for algorithms computing linear transformations, because low condition number promotes the highly desirable properties of robustness and accuracy.

The second extension assumes extra space is allowed, as long as it contains information of bounded norm at the end of the computation. 

The main technical contribution is an extension of matrix  entropy used in Ailon (2013) for unitary matrices to a potential function computable for any invertible matrix, using  ``quasi-entropy'' of ``quasi-probabilities''.
\end{abstract}

\section{Introduction}
The (discrete) normalized Fourier transform is a complex linear mapping sending an input $x\in \C^n$ to $y=Fx\in \C^n$, where $F$ is an $n\times n$ unitary matrix defined by
$$ F(k,\ell) = n^{-1/2}e^{-i2\pi (k-1)(\ell-1)/n}\ .$$
The \emph{unnormalized Fourier transform} matrix is defined as $n^{1/2} F$. \footnote{The unnormalized Fourier transform is sometimes  referred to, in literature, as the ``Fourier transform''.  We prefer to call $F$ the Fourier transform, and $\sqrt n F$ the unnormalized Fourier transform.}
The Fast Fourier Transform (FFT) of Cooley and Tukey \cite{CooleyT64} is a method for computing the Fourier transform  (normalized or not - the adjustment is easy) of a vector $x\in \C^n$
in time $O(n\log n)$ using a so called linear algorithm.  A linear algorithm, as defined in
\cite{Morgenstern:1973:NLB:321752.321761},  is a sequence $\F_0,\F_1,\dots$, where each $\F_i$
is a set of affine functions,  for each  $i\geq 0$ $\F_{i+1} = \F_i \cup \{\lambda_i f + \mu_i g\}$ for some  $\lambda_i, \mu_i\in \C$ and $f,g\in \F_i$, and $\F_0$ contains (projections onto) the input variables as well as constants.

It is trivial  that computing the Fourier Transform requires a linear  number of steps, but 
no non-trivial lower bound is 
known without making very strong assumptions about the computational model.   Papadimitriou, for example, computes in \cite{Papadimitriou:1979:OFF:322108.322118} an $\Omega(n\log n)$ lower bounds for Fourier transforms in finite fields using a notion of an information flow network.   It is not clear how to extend
that result to the Complex field.  There have also been attempts \cite{Winograd76} to reduce the constants hiding in the upper bound
of $O(n\log n)$, while also separately  counting the number of additions versus the number of multiplications (by constants).
In 1973, Morgenstern proved that if the modulus of the $\lambda_i$'s and $\mu_i$'s is bounded by $1$
then the number of steps required for computing the \emph{unnormalized} Fourier transform in the linear algorithm model is at least 
$\frac 1 2 n\log_2 n$.  It should be noted that Cooley and Tukey's unnormalized FFT indeed can be expressed as a linear algorithm
with coefficients of the form $e^{iz}$ for some real $z$, namely, complex numbers of unit modulus.

The main idea of Morgenstern is to define a potential function for each $\F_i$ in the linear algorithm sequence,
equaling the maximal absolute value of a determinant of a square submatrix in a certain matrix corresponding to $\F_i$.
The technical step is to notice that the potential function can at most double in each step.
The determinant of the unnormalized Fourier transform is $n^{n/2}$, hence the lower bound of $\frac 1 2 n\log_2 n$.

The determinant of the \emph{normalized} Fourier transform, however, is $1$.  
 Morgenstern's method can therefore not be used
to derive any useful lower bound for computing the normalized Fourier transform in the linear algorithm model
with constants of at most  unit modulus.  Using constants of modulus $1/\sqrt 2$ in the normalized version of FFT,  on the other hand,
does compute the normalized Fourier transform in $O(n\log n)$ steps.

The normalized and unnormalized Fourier transforms are proportional to each other, and hence we don't believe there should
be a difference between their computational complexities in any reasonable computational model.\footnote{It should also be noted that the determinant of any submatrix of the Fourier matrix has determinant at most $1$. }
%The fact that Morgenstern's method doesn't give us a useful lower bound for the unnormalized Fourier transform is due to the weakness of the model, which does not allow
%us to use constants of modulus that grows with $n$.  If such constants were allowed at any location in the circuit (and not just, say, as a preprocessing
%or postprocessing ``scaling up'' or ``scaling down'' step) then his arguments would break down.   
It is important  to note that, due to the model's weakness, Morgenstern's result
teaches us, upon inspection of the proof, that both matrices $\sqrt n F$ (the unnormalized Fourier transform) and  $\sqrt n \operatorname{Id}$ are in the
same complexity class.  
 More generally, it tells us that all unitary matrices scaled up by the same constant ($\sqrt n$ in this case) are in the same complexity class.  
 Ailon \cite{Ailon13} hence studied  the complexity of the Fourier transform \emph{within} the unitary group.  
In his result he showed that, if the algorithm can only apply $2$ by $2$ unitary transformations at each step, then
at least $\Omega(n\log n)$ steps are required for computing the normalized Fourier transform.
The proof is done by defining a potential function on the matrices $M_i$ defined by composing the first $i$ gates.  The potential
function is simply the Shannon entropies of the probability distributions defined by the squared modulus of elements in the
matrix rows.  (Due to unitarity, each row, in fact, thus defines a probability distribution).
%Note that we are  dealing with classical computation
%and not quantum computation, although it should be noted that the quantum Fourier transform can be performed in $O(\log^2 n)$ operations.

This work takes the idea in \cite{Ailon13} a significant step forward, and obtains a $\Omega(n\log n)$ lower bound
for \emph{any scaling} of the Fourier transform in a stronger model of computation which we call the \emph{uniformly well conditioned}.
At each step, the algorithm can either multiply a variable by a nonzero constant, or perform a unitary transformation involving $2$ variables.  The matrix $M_i$ defining the composition of the first $i$ steps must be well conditioned with constant $R$.  This means that $\|M_i\|\cdot \|M_i^{-1}\| \leq R$,
where $\|\cdot\|$ is spectral norm.
%on a pair of variables.  Additionally, for any $i$, after $i$ steps the matrix $M_i$ defining the transformation 
%obtained
%by the first $i$ steps satisfies that the spectral norm of both $M_i$ and $M^{-1}_i$ is bounded by some number $R$.
Taking this number into account, the actual lower bound we obtain is $\Omega(R^{-1}n\log n)$.
This main result is presented in Section~\ref{sec:mainresult}.
It should be noted that well conditionedness is related to numerical stability:  The less well condtioned
a transformation is, the larger the set of inputs on which numerical errors would be introduced in any
computational model with limited precision.  An important commonly studied example is the linear regression (least squares) problem, in which the condition number controls a tradeoff
between computational complexity and precision \cite{GolubvL}.   
We also note the work of Raz et. al \cite{DBLP:journals/jcss/RazY11}, in which a notion of numerical stability was also used 
to lower bound the complexity of certain functions, although that work does not seem to be directly comparable to this.

Another limitation of \cite{Ailon13} is that no additional memory (extra variables) were allowed in the
computation.  (This limitation is not present in \cite{Morgenstern:1973:NLB:321752.321761}.)
In Section~\ref{sec:space} this limitation is removed, assuming a bound on the amount of information
held in the extra space at the end of the computation.

\subsection{Different Types of Fourier Transforms}

In this work we will assume that $n$ is even and will use $F$ to denote one of the following:
\begin{enumerate}
\item The real orthogonal $n\times n$ matrix computing
the (normalized) complex discrete Fourier transform (DFT) of order $n/2$ on an input $\hat x \in \C^{n/2}$, where the real part of $\hat x$ is stored in $n/2$ coordinates and the imaginary part in the remaining $n/2$. 
\item The (normalized) Walsh-Hadamard Fourier transform, where $n$ is assumed to be an integer power of
two and $F(i,j) = \frac 1 {\sqrt n}(-1)^{\langle [i-1],[j-1]\rangle}$, where for $a\in [0,2^{\log n}-1]$, $[a]$ is the binary vector representing $a$ in base $2$, and $\langle \cdot,\cdot \rangle$ is dot-product over $Z_2$.  It
is well known that $Fx$ given $x\in \R^n$ can be computed in $O(n\log n)$ oeprations using the so called Walsh-Hadamard transform.  All the above discussion on Morgenstern's result applies to this transformation as well.
\end{enumerate}
In fact, the field of harmonic analysis defines a Fourier transforms corresponding to any Abelian group of order $n$, but this abstraction
would not contribute much to the discussion.  Additionally, our results apply to the well known (and useful) cosine transform, which
is a simple derivation of DFT.  In any case,  DFT and Walsh-Hadamard are central to engineering
and the reader is invited to concentrate on those two.

\subsection{A Note on Quantum Computation}
It is important to note that we are in the classical setting, not quantum.  A quantum version of the Fourier transform can be computed
in time $O(\log^2 n)$ using an algorithm by Shor (refer e.g. to Chapter~5, \cite{NielsenC10}) but that setting is different.
%The model of computation which we present in Section~\ref{s} allows the algorithm, in each step, to apply a $2\times 2$ unitary operator on
%two  coordinates of an $n$ dimensional complex vector, serving as both input and output.

% realistic:  If at some step $M_i$ is ill conditioned, then for
%some input (precisely, the eigenvector corresponding to either the largest or smallest eigenvalue) numerical
%errors will be introduced.  

\subsection{A Note on Universality of our Model}
We argue that the model of computation studied in this work is suitable for studying any algorithm
which computes a linear transformation by performing a sequence of simple linear operations. 
Imagine a machine which can perform linear operations acting on $k$ variables in one step, for some constant $k$.
Recall the SVD theorem stating that any such mapping $\psi$ can be written as a composition of three linear mappings, where two are
orthogonal and one is diagonal (multiplication by constants).  Additionally, if $\psi$ is nonsingular,
then all the diagonal constants are nonzero.  Also recall that any orthogonal mapping of rank $k$ can be
decomposed into $O(k^2)$ orthogonal mappings, each acting on at most two coordinates.  Hence, up to a constant
speedup factor, such a machine can be efficiently sumilated using our model.  

Additionally, any fixed precision machine cannot afford extreme ill conditionedness.  The worse the condition
number of the computation is (at some point), the higher the maximal ratio between two numbers the machine
would have to be able to represent.  Otherwise viewed, the worse the condition number, the smaller the set of inputs  for which the computation is guaranteed to be numerically stable.

Hence, this work in fact offers the first tradeoff between running time and precision
of scale-free Fourier transform algorithms in a computational model that is relevant to any
fixed precision machine that performs simple linear operations as atomic steps.
Whether this tradeoff is tight is subject to further investigation.

\subsection{Our Main Technique:    Quasi-Entropy as a Potential Function}
Ailon \cite{Ailon13} defined the entropy of a unitary matrix $M\in \C^{n\times n}$ to be
\begin{equation}\label{entropy} \Phi(M) = \sum_{i=1}^n\sum_{j=1}^n f(M(i,j))\ ,
%-\sum_{p,q} |M(p,q)|^2\log |M(p,q)|^2\ ,
\end{equation}
where for any nonnegative $x$,
\begin{equation}\label{f} f(x) = \begin{cases} 0 & x=0 \\ -|x|^2\log |x|^2  & x> 0\end{cases}\ .
\end{equation}
Since $M$ is unitary, for any row $i$ the numbers $(|M(i,1)|^2,\dots, |M(i,n)|^2)$ form a probability distribution
vector, from which we can view $\Phi(M)$ as the sum of the Shannon entropy of $n$ distributions.
Note that $\Phi(M)$ is always in the range $[0,n\log n]$.  (Throughout, we will take all logarithms
to be in base $2$, as common in information theory). Ailon \cite{Ailon13} claimed, using a simple
norm preservation argument, that for any (complex) Givens matrix $S$, 
\begin{equation}\label{oldAilon}
|\Phi(M) - \Phi(SM)| \leq 2\ ,
\end{equation}
where we remind the reader that a Givens matrix is any unitary transformation acting on two coordinates.
Since $\Phi(\id)=0$ and $\Phi(F) = n\log n$, the conclusion was that at least $\frac 1 2 n\log n$ Givens
operations are required to compute the (normalized) Fourier transformation $F$.

The starting point of this work is extending the definition of $\Phi$ in (\ref{entropy}) to any (nonsingular)
matrix.  Indeed, there is no  reason to believe that an optimal Fourier transform algorithm must be
confined to the unitary group.  Using (\ref{entropy}) verbatim does not help proving a lower bound, 
as one can easily see that $\Phi(M)$ can change by $\Omega(\log n)$ if we multiply a  row of $M$ by a 
nonzero constant $C$ such that  $|C|\neq 1$.  (For example, if a row of $M$ equals $(1/\sqrt{n},\dots, 1/\sqrt{n})$,
then by multiplying the row by $C=2$ additively changes the entropy by $\Omega(\log n)$.)

%The main contribution  of this work is to generalize the entropy function to the group of nonsingular matrices.
We now fix this problem.
For simplicity, we will work over $\R$ and not over $\C$.  The complex Fourier transform can be simulated over
$\R$ by doubling the dimension.\footnote{  This can be done by representing the input (and output) using $2n$ variables, half dedicated to the real
part and half to the imaginary part of the complex input.  Accordingly, each matrix element $F(k,\ell)=n^{-1/2}e^{-i2\pi k\ell/n}$ of the complex Fourier transform  becomes a $2\times 2$ rotation matrix with angle $-2\pi k\ell/n$, multiplied by $n^{-1/2}$.  
}  (Note that over $\R$ unitary matrices are referred to as orthogonal matrices, and we shall
follow this convention.)  For any real nonsingular matrix $M$, we define
\begin{equation}\label{entropy2} 
\Phi(M) := -\sum_{i=1}^n\sum_{j=1}^n \f(M(i,j),\ M^{-1}(j,i))\ ,
\end{equation}
where  for all $x,y\in \R$,
\begin{equation}\label{f2} \f(x,y) := \begin{cases} 0 & x\cdot y=0 \\ -x\cdot y\cdot\log |x\cdot y|  & x\cdot y\neq  0\end{cases}\ .
\end{equation}
Note that if $M$ is orthogonal then $M(i,j)={M^{-1}(j,i)}$.  This implies that $M$ defined in (\ref{entropy2}) is an extension of (\ref{entropy}) from the unitary to the
nonsingular group.  Also note that for all $i$, the numbers $M(i,1)M^{-1}(1,i),\dots, M(i,n)M^{-1}(n,i)$ sum
up to one (by definition of matrix inversion) but they do not form a probability distribution vector because
they may be negative or $>1$ in general, hence we think of them as quasi-probabilities (and of $\Phi$ as quasi-entropy).
Our main Lemma~\ref{mainlem} below shows that a Givens rotation applied to $M$ can change $\Phi(M)$ by at most $O(R)$, where $R$ is the condition number of $M$.
%Note that multiplying a single row of $M$ by nonzero constants clearly does not change $\Phi(M)$ (although it may change its 
%condition number).  Also note that multiplying all rows of $M$ by a nonzero constant (``rescaling'') neither changes $\Phi$ nor the condition number.

\subsection{Contribution and Limitations}

We believe that our main contribution is in showing a lower bound on computation of the
Fourier transform that does not depend on scaling.  Indeed, if  $M_i$ is the matrix
defined by the  composition of the first $i$ steps (see exact model definition below), then $\Phi(M_i)$ is completely insensitive to rescaling  (of all rows) by any arbitrary large or small nonzero number,
because such an operation has no effect on neither $\Phi$ nor on condition numbers.
%Our main claim is, roughly speaking, that the source of complexity inherent in the Fourier transform is in generating high entropy (dense) outputs  from low entropy (sparse) inputs.
We argue that the generalized matrix entropy $\Phi$ defined in (\ref{entropy2}), which is interesting in its own right, is an important key to  understanding the complexity of one of the most important linear transformations used in science and engineering, for which an algorithm that is believed to be optimal has been around for half a decade.  
We point out the following shortcoming of this work, which also gives rise to interesting open problems:
%\begin{itemize}
%\item
Although constant well conditionedness enhances numerical stability,  for many applications it is reasonable to work with condition numbers that grow with $n$, even polynomially.  Our result implies nontrivial bounds for condition number up to $o(\log n)$.  It would be interesting to prove interesting lower bounds for less well conditioned computations.
%\item
%As in \cite{Ailon13} and unlike Morgenstern's \cite{Morgenstern:1973:NLB:321752.321761}, our model of
%computation does now allow the algorithm to use additional
%space.  In Section~\ref{sec:space} we will show how to extend our result to a model in which additional variables are allowed, albeit with a limit on magnitude of the information they may hold at the end of the computation.% value at the end of the computation.
%\end{itemize}

\section{The Well Conditioned Model of Computation}\label{s}

%Before defining our model, we will simplify the analysis by working over $\R$ instead of over $\C$.

%$n$, separating the real and imaginary parts of the variables and the constants.
For a  matrix $M$, we let $M^{(i)}$ denote the $i$'th column of $M$.
Our model of computation consists of layers $L_0,\dots, L_m$, each  containing exactly $n$ nodes and representing
a vector in $\R^n$. The first layer, $L_0 \in \R^n$, is the input.  The last layer $L_m\in \R^n$ is the output.

For $i=1,\dots, m$, the $i$'th gate connects layer $i-1$ with later $i$.
There are two types of gates:  \emph{rotations} and \emph{constants}.
If gate $i$ is a rotation, then there are two indices  $k_i, \ell_i \in [n]$, $k_i < \ell_i$, and an orthogonal matrix $$ A_i = \left ( \begin{matrix} a_i(1,1) & a_i(1,2) \\ a_i(2,1) & a_i(2,2) \end{matrix} \right )\ = \left(
\begin{matrix}
\cos \theta_i & \sin \theta_i \\ -\sin \theta_i & \cos \theta_i \\
\end{matrix}
\right ) .$$
For each $j \not \in\{k_i, \ell_i\}$, $L_i(j) = L_{i-1}(j)$.  The values of $L_i(k_i)$ and $L_i(\ell_i)$ are given as
$$ \left (\begin{matrix} L_i(k_i) \\ L_i(\ell_i) \end{matrix} \right ) = A_i  \left (\begin{matrix} L_{i-1}(k_i) \\ L_{i-1}(\ell_i) \end{matrix} \right )\ .$$

\noindent
Note that the transformation taking $L_{i-1}$ to $L_i$ is known as a \emph{Givens rotation}.

%In words, the next layer is obtained from the current layer by applying a $2$-by-$2$ unitary transformation on two coordinates.

If gate $i$ is of type constant, then it is defined by an index $k_i\in [n]$ and a nonzero $c_i$.  For each $j\neq k_i$, $L_i(j) = L_{i-1}(j)$.  Additionally, $L_i(k_i) = c_i L_{i-1}(k_i)$.
%In words, the next layer is obtained by rescaling the $\ell_i$'th terminal by $c_i$.

\noindent
We will encode the circuit using the sequence
$$ (k_i,\ell_i,\theta_i, c_i)_{i=1}^m\ ,$$
where we formally define $c_i$ to be $0$ for rotation gates, and $\ell_i=0$ for constant gates.

Let $M_i$ be the matrix transforming $L_0$ (as a column vector) to $L_i$.   We say that $M_i$ is the $i$'th
defining matrix of the circuit. If gate $i$ is a rotation, then  $M_i$ is
obtained from $M_{i-1}$ by replacing rows $k_i$ and $\ell_i$ in $M_{i-1}$ by the application of $A_i$ to these rows, stacked one on top of the other to the right of $A_i$.
If gate $i$ is diagonal, then $M_i$ is obtained from $M_{i-1}$ by multiplying row $k_i$ of $M_{i-1}$ by $c_i$.
Also, $M_0 = \Id$.

%Throughout, we will assume that the gate types are alternating: Gate $i$ is diagonal for odd $i$, and unitary for even $i$.  Clearly, any singular latered circuit can be converted to such by at most doubling the number of layers.

%Compared to Morgenstern's model of computation the unitary layered circuit is strictly weaker.  To see why it is not stronger, notice that the matrix elements of $A_i$ all have modulus at most $1$.  It is strictly weaker because it uses only unitary transformations, but also because it has a bounded
%memory of $n$ numbers at any given moment.  Indeed, it is not possible in layer $L_{i+1}$ to use a coordinate of $L_{i'}$
%for $i' < i$.  Still, the normalized FFT is implemented as a unitary layered circuit with $m=O(n\log n)$.

\begin{defn}
A layered circuit of depth $m$  is $R$-uniformly well conditioned (for some $R > 1$)
if  $$\max_{i\in[m]}\{ \|M_i\|\cdot \|M^{-1}_i\|\} \leq R\ .$$
%where $\|\cdot\|$ is spectral norm.
\end{defn}

Note that a $1$-uniformly well conditioned circuit recovers the model of \cite{Ailon13} (restricted over the reals).

%It is useful to note that for all $i'$, the matrix $(M_{i'}^{-1})^T$ is obtained by replacing, for all $i\leq i'$ such that the
%$i$'th gate is diagonal, the constant $c_i$ with $1/c_i$.  This will be used in what follows.

\section{The Main Result}\label{sec:mainresult}
\begin{thm}\label{thm:main}
If an $R$-uniformly well conditioned layered circuit  $\CT=(k_i,\ell_i,\theta_i,c_i)_{i=1}^m$ computes a transformation that is proportional to the Fourier transform $F$, then  the number of rotations is $\Omega(R^{-1}n\log n)$.
%Moreover, the same lower bound counts the number of rotations only
\end{thm}

%We will prove the theorem assuming that the circuit computes the normalized Fourier transform $F$ exactly.  It will be  obvious upon inspection of the proof that normalization is not important.

\begin{proof}
%For a matrix $M$ and a set $I\subseteq [n]$ of indices, let $M[I]$ denote the principal minor corresponding to the set $I$.
%For $i=1,\dots,m$ let $\tilde A_i$ denote the matrix defined so that $\tilde A_i[\{k_i,\ell_i\}] =  A_i$, $\tilde A_i[[n]\setminus\{k_i,\ell_i\}] = \operatorname{Id}$ and $\tilde A_i(p,q) = 0$ whenever exactly one of $p,q$ is in $\{k_i,\ell_i\}$.
%It is clear that
%$$ L_i = \tilde A_i \tilde A_{i-1} \cdots \tilde A_1 L_0\ .$$
%It hence makes sense to define $M_i = \tilde A_i \tilde A_{i-1} \cdots \tilde A_1$.  Note that $M_m = F$, where $F$ is the normalized
%FFT matrix.  We also define $M_0 = \Id$.
%For a unitary matrix $M$, Ailon \cite{Ailon13} defined  a potential function
%\begin{equation}\label{entropy} \Phi(M) = -\sum_{p,q} |M(p,q)|^2\log |M(p,q)|^2\ ,
%\end{equation}
%where  $0\log 0$ was formally defined to be $\lim_{x\rightarrow 0^+} x\log x = 0$.
%The key lemma in \cite{Ailon13} was the following:  If the circuit contains no diagonal gates, then  or all $i\in[m]$,
%\begin{equation}\label{key}\left | \Phi(M_i) - \Phi(M_{i-1})\right |\leq 2\ .\end{equation}

%The goal is to extend the potential function $\Phi$ to nonsingular, well-conditioned matrices.
%Unfortunately Using (\ref{entropy}) verbatim is not useful for the purpose of obtaining an $\Omega(n\log n)$
%lower bound.  Instead, we work with the following:

%Note that if $M$ is unitary, then $(M^{-1})^T=M$ and, hence, definition (\ref{entropy2}) coincides with (\ref{entropy}).
%We now analyze to what extent  $\Phi(M_i)$ can differ from $\Phi(M_{i-1})$.  

We begin with an observation, which can be
proven with a simple induction: For any $i\in[m]$, $(M_i^{-1})^T$ is the $i$'th defining matrix of a circuit $\CT'$
defined by $(k_i,\ell_i,\theta_i,c_i')_{i=1}^m$, where $c_i'=1/c_i$ if the $i$'th gate of $\CT$ is of type constant,
and $0$ otherwise.
A clear consequence of this observation is that if the $i$'th gate of $\CT$ is of type constant, then
$$ \Phi(M_{i-1}) = \Phi(M_i)\ .$$
Indeed, just notice that for $p=k_i$ and any $q\in[n]$, $M_i(p,q) = c_i M_{i-1}(p,q)$ and $(M_i^{-1})^T(p,q) = c_i^{-1}(M_{i-1}^{-1})^T(p,q)$.
We analyze the effect of rotation  gates on $\Phi$.  To this  end, we need the following lemma.
\begin{lem}\label{mainlem}
Recall $\f$ as in (\ref{f2}).  For $4$ real numbers $w,x,y,z$, define
%Let $f(x) = -x\log|x|$.
%For a $2$ by $2$ matrix $Z$, let 
$$\Psi(w,x,y,z) =  \f(w,x)+ \f(y, z)\ .$$
\noindent
Now define
\begin{eqnarray*}\label{eq:mainlem}
\alpha(w,x,y,z) &=& \sup_{\theta\in [0,2\pi]} \Psi(w\cos \theta + y\sin \theta, x\cos \theta + z\sin \theta, -w\sin \theta + y \cos \theta, -x\sin \theta + z\cos \theta) \\
\beta(w,x,y,z) &=& \inf_{\theta\in [0,2\pi]} \Psi(w\cos \theta + y\sin \theta, x\cos \theta + z\sin \theta, -w\sin \theta + y \cos \theta, -x\sin \theta + z\cos \theta)\ . \\
\end{eqnarray*}
Then
\begin{equation}\label{mainlemeq}
\sup_{w,x,y,z} \frac{\alpha(w,x,y,z)- \beta(w,x,y,z)}{\sqrt{(w^2+y^2)(x^2+z^2)}}= O(1)\ .
\end{equation}
(we formally define the last fraction as $0$ if either $w^2+y^2=0$ or $x^2+z^2=0$.  Note that in this degenerate case both $\alpha(w,x,y,z)=0$ and $\beta(w,x,y,z)=0$ ).  %$(w,x,y,z)=(0,0,0,0)$.)
\end{lem}

The proof of the lemma is deferred to Section~\ref{sec:proofmainlem}. Now let  $i$ be such that
the $i$'th gate is a rotation. Then using the definition of $\Psi$ as in the lemma,
\begin{eqnarray}
\Phi(M_i) - \Phi(M_{i-1}) &=& \sum_{q=1}^n \left [\Psi(M_i(k_i,q),\, M_i^{-1}(q, k_i),\, M_i(\ell_i,q),\, M_i^{-1}(q,\ell_i)) \right . \nonumber \\
& &  -  \left . \Psi(M_{i-1}(k_i,q),\, M_{i-1}^{-1}(q,k_i),\, M_{i-1}(\ell_i,q),\, M_{i-1}^{-1}(q,\ell_i)) \right] \nonumber
\end{eqnarray}
By Lemma~\ref{mainlem}, hence for some global $C>0$
\begin{eqnarray}
%|\Phi(M_i) - \Phi(M_{i-1})| \leq \sum_{q=1}^n \left  (M_i(k_i, q)^2 + (M_i^{-1})^T(k_i,q)^2 + M_i(\ell_i, q)^2 + (M_i^{-1})^T(\ell_i,q)^2 \right )\ . \nonumber
|\Phi(M_i) - \Phi(M_{i-1})| &\leq& C\sum_{q=1}^n \sqrt{\left( M_i(k_i, q)^2 +M_i(\ell_i, q)^2\right) \left ( M_i^{-1}(q,k_i)^2 + M_i^{-1}(q,\ell_i)^2\right )} \nonumber \\
&\leq &  C \sqrt{ \left (\sum_{q=1}^n M_i(k_i, q)^2 +M_i(\ell_i, q)^2  \right ) \left ( \sum_{q=1}^n  M_i^{-1}(q,k_i)^2 + M_i^{-1}(q,\ell_i)^2\right )} \nonumber \\
&\leq& 2C \|M_i\|\cdot \|M_{i}^{-1}\| \leq 2CR\ ,
%
%\sum0_{q=1}^n \sqrt{\left( M_i(k_i, q)^2 +M_i(\ell_i, q)^2\right) \left ( (M_i^{-1})^T(k_i,q)^2 + %(M_i^{-1})^T(\ell_i,q)^2\right )}\ .\nonumber
% \left  (M_i(k_i, q)^2 + (M_i^{-1})^T(k_i,q)^2 + M_i(\ell_i, q)^2 + (M_i^{-1})^T(\ell_i,q)^2 \right )\ . \nonumber
\end{eqnarray}
where the second inequality is Cauchy-Schwarz, and the third is from the definition of condition number (together with the observation that the norm of any row or column of a matrix is at most the spectral norm of the matrix).
Hence,
%\noindent
%By the assumption of uniform condition number, hence,
\begin{eqnarray}
|\Phi(M_i) - \Phi(M_{i-1})| \leq O(R)\ .
\end{eqnarray}

Now notice that $\Phi(M_0) = \Phi(\Id) = 0$ and $\Phi(M_m) = \Phi(F) = n\log n$.
Hence $m = \Omega(R^{-2}n\log n)$, as required.

%\noindent
%It remains to prove Lemma~\ref{mainlem}.  This is done using standard analytic tools.
\end{proof}

\section{Proof of Lemma~\ref{mainlem}}\label{sec:proofmainlem}
If either $(w,y)=(0,0)$ or $(x,z)=(0,0)$ then the LHS of (\ref{mainlemeq}) is clearly $0$.
Assume first that the vectors $(w,y)$ and $(x,z)$ are not proportional to each other.
Without loss of generality, we can assume that the vector direction $(1,0)\in \R^2$ is an angle bisector of the two segments
connecting the origin with $(w,y)$ and $(x,z)$.  In words, there exist numbers $r,s> 0$ and an angle $\phi$
such that
\begin{eqnarray*} 
(w,y) &=& \left (r\cos \frac \phi 2, r\sin \frac  \phi 2\right) \\
(x,z) &=& \left (s\cos \frac \phi 2, -s\sin \frac \phi 2\right)\ . \\
\end{eqnarray*}
By symmetry, we can assume that $\phi\in[-\pi/2,\pi/2]\setminus\{0\}$, because otherwise we could replace $w$ with $-w$ and $y$ with $-y$, which would result in negation of $\Psi$ (leaving $(\alpha-\beta)$ untouched).
In fact, we can assume that $\phi\in(0,\pi/2]$, because otherwise we would replace the roles of $(w,y)$ and $(x,z)$).
%Assume  first that  $\phi \leq \pi/2$.  
With this notation, we have for all $\theta\in [0,2\pi)$
\begin{eqnarray}
w\cos\theta + y\sin \theta = r\cos\left(\frac \phi 2 + \theta\right )  & &
x\cos\theta + z\sin \theta = s\cos\left(-\frac \phi 2 + \theta\right ) \\
-w\sin\theta + y\cos \theta = r\sin\left(\frac \phi 2 + \theta\right ) & &
-x\sin\theta + z\cos \theta = s\sin\left(-\frac \phi 2 + \theta\right ) 
\end{eqnarray}
Therefore, 
\begin{eqnarray}
\Psi(w\cos \theta + y\sin \theta, x\cos \theta + z\sin \theta, -w\sin \theta + y \cos \theta, -x\sin \theta + z\cos \theta)  = \nonumber
\end{eqnarray}
\begin{eqnarray}
  -rs \cos\left(\frac \phi 2 + \theta\right )\cos\left(-\frac \phi 2 + \theta\right ) \log \left |rs \cos\left(\frac \phi 2 + \theta\right )\cos\left(-\frac \phi 2 + \theta\right ) \right | \   \label{row1} \\
   -rs \sin\left(\frac \phi 2 + \theta\right )\sin\left(-\frac \phi 2 + \theta\right ) \log \left |rs \sin\left(\frac \phi 2 + \theta\right )\sin\left(-\frac \phi 2 + \theta\right ) \right | .  \label{row2} 
\end{eqnarray}
We view the last expression as  a function of $\theta$, and write $\Psi(\theta)$ for shorthand.
The function $\Psi$ is differentiable everywhere except $\theta \in Q =  \{ \pm \frac\phi 2 + j\frac \pi 2\}$ for $j=0,1,2,\dots$. For $\theta \in Q$, it is not hard to see that $\Psi$ is not a local optimum.  It hence suffices to find local optima of $\Psi$ for $\theta \not \in Q$.  
Consider first the range $\theta \in \left ( -\frac \phi 2, \frac \phi 2\right )$. In this range, the argument inside the absolute value in (\ref{row1}) is positive, while the one inside (\ref{row2}) is negative.  Differentiating with respect to $\theta$, we get
\begin{eqnarray}
\frac {d }{d \theta} \Psi(\theta)  &=& rs\left [  \sin\left(\frac \phi 2 + \theta\right )\cos\left(-\frac \phi 2 + \theta\right ) +  \cos\left(\frac \phi 2 + \theta\right )\sin\left(-\frac \phi 2 + \theta\right ) \right ]  \nonumber\\
& &\hspace{4cm} \times \left [1+ \log\left ( rs \cos\left(\frac \phi 2 + \theta\right )\cos\left(-\frac \phi 2 + \theta\right ) \right ) \right ] \nonumber \\
&+& rs\left [ \cos\left(\frac \phi 2 + \theta\right )\sin\left(-\frac \phi 2 + \theta\right )+\sin\left(\frac \phi 2 + \theta\right )\cos\left(-\frac \phi 2 + \theta\right ) \right ] \nonumber \\
& &\hspace{4cm} \times \left [1- \log \left (-rs \sin\left(\frac \phi 2 + \theta\right )\sin\left(-\frac \phi 2 + \theta\right ) \right )\right ]  \nonumber \\
&=& rs(\sin 2\theta )\left  (2 + \log \left (-\frac{\cos\left(\frac \phi 2 + \theta\right )\cos\left(-\frac \phi 2 + \theta\right )}{\sin\left(\frac \phi 2 + \theta\right )\sin\left(-\frac \phi 2 + \theta\right ) }\right )\right ) \ .   \nonumber
\end{eqnarray}

One checks using standard trigonometry that the last $\log$ is nonnegative.  Hence, $d \Psi/d \theta$ vanishes
only when $\theta=0$.  For this value, 
\begin{equation}\label{ext1} \Psi(0) = -rs \cos^2 \left (\frac \phi 2 \right )\log\left ( r s \cos^2\left ( \frac \phi 2 \right)\right) +rs \sin^2 \left( \frac \phi 2 \right )\log\left ( rs\sin^2 \left (\frac \phi 2\right )\right )\ .
\end{equation}

We now study the case $\theta \in (\phi/2, \pi/2-\phi/2)$.
In this range, the argument inside the absolute values in both (\ref{row1}) and  (\ref{row2}) is positive.
For this case, using a similar derivation as above, the derivative $d\Psi/d\theta$ equals
\begin{eqnarray}
\frac {d }{d \theta} \Psi(\theta)  &=& 
 rs(\sin 2\theta ) \log \left (\frac{\cos\left(\frac \phi 2 + \theta\right )\cos\left(-\frac \phi 2 + \theta\right )}{\sin\left(\frac \phi 2 + \theta\right )\sin\left(-\frac \phi 2 + \theta\right ) }\right ) \ .   \nonumber
\end{eqnarray}

By our assumption on $\phi$, the last derivation vanishes eg when $\theta = \pi/4$.
For this value,
\begin{eqnarray}
\Psi(\pi/4) &=& -rs \cos\left (\frac \phi 2 + \frac \pi 4\right )\cos\left(-\frac \phi 2 + \frac \pi 4\right ) 
\log \left ( r s  \cos\left (\frac \phi 2 + \frac \pi 4\right )\cos\left(-\frac \phi 2 + \frac \pi 4\right )\right ) \nonumber \\
& & -rs \sin\left (\frac \phi 2 + \frac \pi 4\right )\sin\left(-\frac \phi 2 + \frac \pi 4\right ) 
\log \left ( r s  \sin\left (\frac \phi 2 + \frac \pi 4\right )\sin\left(-\frac \phi 2 + \frac \pi 4\right )\right ) \nonumber 
\end{eqnarray}
By basic trigonometry, one verifies that 
\begin{eqnarray*} \cos\left (\frac \phi 2 + \frac \pi 4\right )\cos\left(-\frac \phi 2 + \frac \pi 4\right )  &=&
\frac 1 2 \cos^2\left (\frac \phi 2\right ) - \frac 1 2 \sin^2\left (\frac \phi 2\right ) = \frac 1 2 \cos \phi =  \\
& &\hspace{6cm}\sin\left (\frac \phi 2 + \frac \pi 4\right )\sin\left(-\frac \phi 2 + \frac \pi 4\right )  \nonumber 
% &=&
%\frac 1 2 \cos^2\left (\frac \phi 2\right ) - \frac 1 2 \sin^2\left (\frac \phi 2\right )\ .
\end{eqnarray*}
Plugging in our derivation of $\Psi(\pi/4)$, we get
\begin{eqnarray}
\Psi(\pi/4) &=& -rs\left (  \cos^2\left (\frac \phi 2\right ) -  \sin^2\left (\frac \phi 2\right )\right )
\log \left ( r s  \cos\left (\frac \phi 2 + \frac \pi 4\right )\cos\left(-\frac \phi 2 + \frac \pi 4\right )\right )\ . \nonumber \\
%& & -\frac 1 2rs\left ( \cos^2\left (\frac \phi 2\right ) -  \sin^2\left (\frac \phi 2\right )\right )
%\log \left ( r s  \sin\left (\frac \phi 2 + \frac \pi 4\right )\sin\left(-\frac \phi 2 + \frac \pi 4\right )\right )\ . \nonumber 
\end{eqnarray}

It is not hard to verify that $\Psi(0)$ and $\Psi(\pi/4)$ are the only extremal values of $\Psi$.  Now notice that
in the expression $\left |\Psi(\pi/4) - \Psi(0)\right |$, the term $\log (rs)$ is cancelled out, and we are left 
with $|\Psi(\pi/4) - \Psi(0)| = rs g(\phi)$, where
\begin{eqnarray}
& & g(\phi) =  \nonumber  \\
& & \ \ \left | \cos^2 \left (\frac \phi 2\right ) \log \left (\frac{\cos^2\left(\frac \phi 2\right)}{\cos\left (\frac \phi 2 + \frac \pi 4\right )\cos\left (-\frac \phi 2 + \frac \pi 4\right )}\right )
+
\sin^2 \left (\frac \phi 2\right ) \log \left (\frac{\cos\left (\frac \phi 2 + \frac \pi 4\right )\cos\left (-\frac \phi 2 + \frac \pi 4\right )}{\sin^2\left(\frac \phi 2\right)}\right )
 \right |\ . \nonumber \\
 &=& \left | \cos^2 \left (\frac \phi 2\right )  \log \cos^2 \left (\frac \phi 2\right )   -\sin^2 \left (\frac \phi 2\right )  \log \sin^2 \left (\frac \phi 2\right ) -\cos \phi \log \frac{\cos \phi}{2} \right |\ .\ \ \
\end{eqnarray}
The function $g(\phi)$ is bounded in the range $\phi\in(0,\pi/4]$, by which we conclude that 
for some global constant $C$,
$$ \left | \sup_\theta \Psi(\theta) - \inf_\theta \Psi(\theta) \right | \leq Cr s \ .$$

This concludes the proof for the case $\phi\not\in \{0,\pi\}$ (modulo $2\pi$). 
If $(w,y)$ and $(x,z)$ are proportional to each other ($\phi\in\{0,\pi\}$), then the analysis uses the same
simple norm preservation argument as in \cite{Ailon13}.  (Details omitted)
% Other cases are argued similarly, we omit the details.
%This concludes the proof of the lemma.

\section{Using Additional Space}\label{sec:space}
We assume in this section that aside from the $n$ input variables, the algorithm has access to an additional
memory of total size $N$.  We would like to explore to what extent this additional memory could help in
Fourier computation, using the framework developed in the previous sections.
We will assume throughout that
\begin{equation}\label{boundN} N \leq n\log n\ ,\end{equation}
because, assuming all extra memory is accessed in the linear circuit, $N/2$ is a  lower bound on the depth
of the circuit.  Additionally, we will assume that this additional memory is initialized as $0$.  This is not
a real restriction, because the Fourier transform is a homogenous transformation.\footnote{More precisely, if require
to initialize a subset of the additional memory with values $\neq 0$, then  its total linear contribute to the output variables
must be $0$.}
For convenience, we will work with linear circuits as defined in Section~\ref{s} over $\R^{n+N}$, and
%assume the the $N$ extra input variables are initialized to $0$.  Note that this assumption is not
%restrictive, because the linear contribution of the  extra  $N$ input variables to the $n$ output variables
%must be null.

As a warmup, we will also assume that the $N$ \emph{extra output variables} are identically $0$.
In other words, that there is no ``garbage information'' in the $N$ output variables.
%w.l.o.g. that both the input and the output consist of the first $n$ coordinates before and after
%the computation, respectively.  
This means that, if the circuit has depth $m$  
then $[M_m]_{[n],[n]}=F$  and $[M_m]_{[n+N]\setminus[n],n]}=0$, where for a matrix $A$ and integer sets $I,J$, $[A]_{I,J}$ denotes the submatrix of $A$ corresponding
to rows $I$ and columns $J$.  We will later relax this assumption.
% As before, the matrix $M_i$ of size $(n+N)\times(n+N)$ will encode the computation  after $i$ gates are executed.  F

\begin{thm}\label{thm:partial}
If an $R$-uniformly well conditioned layered circuit $\CT$ computes a transformation $M$ such
that $M_{[n],[n]}=F$ and $M_{[n+N]\setminus[n], [n]}=0$, then the number of rotations in the circuit
is $\Omega(R^{-1}n\log n)$.
\end{thm}
\begin{proof}
We proceed as in the proof of Theorem~\ref{thm:main}, except we now work with a partial entropy function
defined as follows:
\begin{equation}\label{entropypartial} 
\Phi_n(M) := -\sum_{i=1}^{n+N}\sum_{j=1}^n \f(M(i,j),\ M^{-1}(j,i))\ ,
\end{equation}
It is easy to see that, as before, for any $i$ such that the $i'th$ gate is a rotation, $|\Phi_n(M_i) - \Phi_n(M_{i+1})| = O(R)$.
We also notice that, by the assumptions, we must have  $[M_m^{-1}]_{[n],[n]}=F^{-1}$ and $[M_m^{-1}]_{[n],[N+n]\setminus[n]}=0$.  This implies, as before, that $\Phi_m(M_0)=0$ and $\Phi_m(M_m) = \Omega(n\log n)$, leading
to the claimed result.
\end{proof}

It is arguably quite restrictive to assume that the extra space must be clean of any ``garbage'' information
  at the end of the computation.   In particular, by inspection of the last proof,
 the ``garbage'' could have a negative contribution to $\Phi_n$, possibly reducing the computational lower bound.
This assumption is relaxed in what follows.
We will first need a technical lemma.
\begin{lem}\label{lemma2appendixA}
There exists a global constant $C_0<1/4$ such that the following holds for all $n$.
Let $\eps\in\R^n$ be such that $\|\eps\|_2  \leq C_0$.  Then 
\begin{equation}\label{lemma2appendixAeq} -\sum_{i=1}^n \frac 1{\sqrt n}\left ( \frac 1 {\sqrt n} + \eps_i\right ) \log\left |  \frac 1{\sqrt n}\left ( \frac 1 {\sqrt n} + \eps_i\right )\right | \geq \frac 3 4 \log n \ .
\end{equation}
\end{lem}
The proof of the lemma is deferred to Appendix~\ref{appendixproof}.
We are now ready to state and prove the main result in the section.
We will prove the result for the Walsh-Hadamard Fourier transform for simplicity only, although a technical
extension of the last lemma can be used to prove a similar result for any Fourier transform.
\begin{thm}\label{thm:partial2}
Assume $F$ is the Walsh-Hadamard matrix.
Let $\CT$ be an $R$-uniformly well conditioned layered circuit of depth $m$.
Assume that $[M_m]_{[n],[n]}=F$ and that additionally the spectral norm of
$[M_m]_{[N+n]\setminus[n],[n]}$ is at most $C_0/R$, where $C_0$ is from Lemma~\ref{lemma2appendixA}.
  Then $m =\Omega(R^{-1}n\log n)$.
\end{thm}

Note that in both Theorems~\ref{thm:partial} and~\ref{thm:partial2} the normalization chosen in the theorems is immaterial.  For example, we could have replaced $F$ and $C_0/R$ in 
Theorem~\ref{thm:partial2} with $c\cdot F$ and $c\cdot C_0/R$, respectively, for any nonzero $c$.  Hence these lower bounds are insensitive to scaling.  We chose the specific normalization  to eliminate extra constants in the analysis.
\begin{proof}
We will work with the potential function $\Phi_n$ defined in (\ref{entropypartial}), except
that we cannot know the exact value of $\Phi_n(M_m)$ as in the proof of Theorem~\ref{thm:partial}.
Denote the columns of $[M_m]_{[N+n]\setminus[n],[n]}$ by $u_1,\dots,u_n \in \R^{N}$,
 and the columns
of $\left [(M^{-1}_m)^T\right ]_{[N+n]\setminus[n],[n]}$ by $v_1,\dots, v_n\in  \R^N$. 
By the bound on the spectral norm of $[M_m]_{[N+n]\setminus[n],[n]}$, we have in particular a uniform bound on its column
norms:
\begin{equation}\label{eqaa}
\max\{\|u_1\|,\dots, \|u_n\|\} \leq 1/(4R)\ .
\end{equation}
  Hence the norm of any of the first $n$ columns
of $M_n$ is in the range $[1,1+1/(4R)]$.  This implies that the spectral norm $\|M_m\|$ is at least $1$.  By the well conditionedness of $M_m$, we conclude that $\|M^{-1}_m\|$ is at most $R$, by which we conclude that
\begin{equation}\label{eqbb}\max\{\|v_1\|,\dots, \|v_n\|\} \leq R\ .
\end{equation}  

Using Lemma~\ref{lemmaappendixA} in Appendix~\ref{appendixA} together with  the constraints (\ref{eqaa}) and (\ref{eqbb}), we have that for any $j\in [n]$,
\begin{equation}\label{eq00} \sum_{i=1}^{N} \f(u_j(i), v_j(i)) \geq -\frac 1 4 \log 4 - \frac 1 4 \log N= -\frac 1 2 - \frac 1 4 \log N \geq -\frac 1 2 -\frac 1 2 \log n\ , \end{equation}
where in the rightmost inequality we used the assumption that $N\leq n \log n\leq n^2$.

We now need to lower bound the contribution of the upper left square of $M_m$ to the total entropy, namely $\sum_{i,j=1}^n \f([M_m](i,j), [M_m]^{-1}(j,i))$.   By definition if matrix inverse,
$$ [M_m]^{-1}_{[n],[n]} F + [M_m]^{-1}_{[n],[N+n]\setminus [n]} [M_m]_{[N+n]\setminus[n],[n]} = \Id_n\ .$$
Hence,
$$ [M_m]^{-1}_{[n],[n]}  = F - [M_m]^{-1}_{[n],[N+n]\setminus [n]} [M_m]_{[N+n]\setminus[n],[n]} F\ .$$
Letting $\errmat$ denote the error term $-[M_m]^{-1}_{[n],[N+n]\setminus [n]} [M_m]_{[N+n]\setminus[n],[n]} F$,
we can succinctly write $$ [M_m]^{-1}_{[n],[n]} = F+\errmat$$
 and then use the norm chain rule to bound:
\begin{eqnarray}
 \|\errmat\| &\leq& \|[M_m]^{-1}_{[n],[N+n]\setminus [n]}\|\cdot\| [M_m]_{[N+n]\setminus[n],[n]}\|\cdot\| F\| \\
&\leq&  \|[M_m]^{-1}_{[n],[N+n]\setminus [n]}\| \cdot (C_0/R) \cdot 1\ ,
\end{eqnarray}
where we used the  spectral norm bound assumption from the theorem statement.  To bound $\|[M_m]^{-1}_{[n],[N+n]\setminus [n]}\|$, note that
trivially $\|[M_m]^{-1}_{[n],[N+n]\setminus [n]}\| \leq \|[M_m]^{-1}\|$, and $\|[M_m]^{-1}\|$ is  bounded by $R$ (because $\|[M_m]\|\geq \|F\|= 1$ and $[M_m]$ is $R$-well conditioned).  Hence,
\begin{equation}\label{errmatbound}
\|\errmat\| \leq C_0\ .
\end{equation}

%\noindent
The last inequality also implies  that any column of $[M_n]^{-1}_{[n],[n]}$ has norm at most $C_0$.
Using Lemma~\ref{lemma2appendixA} for each  $i\in[n]$, we get
\begin{equation}\label{gkgk}
 \sum_{j=1}^n \f([M_m]^{-1}(j,i), [M_m](i,j)) \geq \frac 1 4 \log n\ .
\end{equation}
(Note that to be precise, to use Lemma~\ref{lemma2appendixA} we need to  flip the sign of $\errmat(j,i)$ whenever $F(i,j)$ is negative, but this is a small technicality.)
%Now let $h_1,\dots, h_n\in \R^n$ denote the columns of $[M_m]_{[n],[n]}=F$, and
%$z_1,\dots, z_n\in \R^n$  the columns of $\left[ (M_m^{-1})^T\right]_{[n],[n]}$.
%We have that  for all $j\in[n]$,
%\begin{equation}\label{eqcc} \max_j\{\|h_j\|_\infty\} \leq \sqrt{2/n}\ \ \ \ \ \ \max_j\{\|z_j\|_2\} \leq R\ ,
%\end{equation}
%where the left inequality is by properties of the normalized Fourier transform $F$ and
%the right one is
%by the assumption of well conditionedness of $M_m$.  By definition of matrix inverse, we have that
%for all $j\in n$,
%$$ z_j^t h_j  + v_j^t u_j= 1\ .$$
%But by Cauchy Schwarz and (\ref{eqaa}) and (\ref{eqbb}), $|v_j^t u_j|\leq 1/4$.  Hence,
%\begin{equation}\label{eqdd} 3/4\leq z_j^t h_j \leq 5/4\ .\end{equation}
%Combining (\ref{eqcc}) and (\ref{eqdd}) we conclude by using standard optimizion techniques (see Appendix~\ref{appendixB}) that
%\begin{equation}\label{eqee} \sum_{i=1}^n \f(u_j(i), v_j(i)) \geq \frac 3 4 \log n - \frac 3 4(1+\log R)\ .
%\end{equation}
Combining (\ref{eq00}) and (\ref{gkgk}) we conclude that

$$ \Phi_n(M_m) \geq  \frac 1 4 n\log n - \frac 1 2 n\ .$$
Since $\Phi_n(M_0)=0$ and $|\Phi_n(M_t) - \Phi_n(M_{t-1})|=O(R)$ for all $t>1$,
we conclude that $m =\Omega(R^{-1}n\log n)$ as required.
%Since $\Phi_n(M_0)=0$,
\end{proof}

\subsection{Future Work}
We raise the following three questions:  \begin{enumerate}
\item Is the dependence in $R$ in Theorem~\ref{thm:main} tight?  Intuitively we did not take full advantage of the
well conditionedness property.  Indeed, we only used it to bound the product of the norm of a row in a matrix and the corresponding row in its inverse.
\item Is the dependence of the lower bound on the norm of columns of the ``additional space at output'' in $R$ in Theorem~\ref{thm:partial2} tight?
\item 
In Section~4.5.4 in \cite{NielsenC10}, it is shown that most (with respect to the H\"aar measure) unitary operators require $\Omega(n^2)$ steps using $2\times 2$ unitary operations.  This is also true if we relax unitarity and allow well-conditioned circuits.  This implies that the techniques developed here cannot be used to prove tight lower bounds for most unitary matrices,  because the potential of unitary matrices is globally upper bounded by $O(n\log n)$.
Nevertheless, we ask: Can $\Phi(M)$ defined in this work  be used for proving lower bounds for other interesting linear algebraic operations? %  linear algebraic lower bounds?  
\end{enumerate}

\bibliographystyle{plain}
\bibliography{low_bound_fft}

\begin{thebibliography}{1}

\bibitem{Ailon13}
Nir Ailon.
\newblock A lower bound for fourier transform computation in a linear model
  over 2x2 unitary gates using matrix entropy.
\newblock {\em Chicago J. of Theo. Comp. Sci.}, 2013.

\bibitem{NielsenC10}
Isaac~L. Chuang and Michael~A. Nielsen.
\newblock {\em Quantum Computation and Quantum Information}.
\newblock Cambridge University Press, 2010.

\bibitem{CooleyT64}
J.~W Cooley and J.~W Tukey.
\newblock An algorithm for the machine computation of complex {F}ourier series.
\newblock {\em J. of American Math. Soc.}, pages 297--301, 1964.

\bibitem{GolubvL}
Gene~H. Golun and Charles~F. van Loan.
\newblock {\em Matrix Computations}.
\newblock The Johns Hopkins University Press, 2 edition, 1989.

\bibitem{Morgenstern:1973:NLB:321752.321761}
Jacques Morgenstern.
\newblock Note on a lower bound on the linear complexity of the fast {F}ourier
  transform.
\newblock {\em J. ACM}, 20(2):305--306, April 1973.

\bibitem{Papadimitriou:1979:OFF:322108.322118}
Christos~H. Papadimitriou.
\newblock Optimality of the fast {F}ourier transform.
\newblock {\em J. ACM}, 26(1):95--102, January 1979.

\bibitem{DBLP:journals/jcss/RazY11}
Ran Raz and Amir Yehudayoff.
\newblock Multilinear formulas, maximal-partition discrepancy and mixed-sources
  extractors.
\newblock {\em J. Comput. Syst. Sci.}, 77(1):167--190, 2011.

\bibitem{Winograd76}
S.~Winograd.
\newblock On computing the discrete {F}ourier transform.
\newblock {\em Proc. Nat. Assoc. Sci.}, 73(4):1005--1006, 1976.

\end{thebibliography}

\appendix
\section{A Useful  Quasi-Entropy Bound}\label{appendixA}

\begin{lem}\label{lemmaappendixA}
Let $x\in \R^n, y\in \R^n$ be two vectors such that $\|x\|_2=\alpha, \|y\|_2=\beta$.
Then
\begin{equation}
-\alpha\beta\log n - |\alpha\beta\log\alpha\beta| \leq \sum_{i=1}^n \hat f(x_i, y_i) \leq -\alpha \beta \log n + |\alpha\beta\log\alpha\beta|\ .
\end{equation}
\end{lem}
\begin{proof}
Notice that
$$ \sum_{i=1}^n \hat f(x_i,y_i) = \alpha\beta \sum_{i=1}^n \hat f(x_i/\alpha, y_i/\beta) - \sum_{i=1}^n x_i y_i\log|\alpha \beta|\ .$$
Since $|\sum x_i y_i| \leq \alpha \beta$ by Cauchy-Schwartz, it hence suffices to prove the lemma for the case $\alpha=\beta=1$,
which we assume from now on.
Let $F(x,y) = \sum_{i=1}^n \hat f(x_i, y_i)$.  Then
$$\nabla F(x,y) = (-y_1(\log|x_1y_1|+1)\dots -y_n(\log |x_1y_1|+1),-x_1(\log|x_1y_1|+1)\dots-x_n(\log|x_ny_n|+1))\ .$$
The gradients of the constraints $G_1(x,y)=\|x\|^2$ and $G_2(x,y)=\|y\|^2$ are:
\begin{eqnarray*}
\nabla G_1(x,y)&=(2x_1\dots2x_n,0\dots 0) \\
\nabla G_2(x,y)&=(0\dots 0, 2y_1\dots2y_n)\ .
\end{eqnarray*}
\end{proof}
Using standard optimization principles, any optima $(x,y)$ of $F$ under $G_1=G_2=1$  satisfies that there exists $\lambda_1,\lambda_2$ such that $\nabla F(x,y) = \lambda_1 \nabla G_1(x,y) + \lambda_2 \nabla G_2(x,y)$.
This implies that for all $i\in [n]$, $x_i/y_i=\lambda_1(\log|x_iy_i|+1)$ and $y_i/x_i=\lambda_2(\log|x_iy_i|+1)$,
hence $x_i^2/y_i^2 = \lambda_1/\lambda_2$.  But we assumed that $\|x\|^2=\|y\|^2=1$, hence $\lambda_1=\lambda_2$
and for all $i$ either $x_i=y_i$ or $x_i=-y_i$.

It remains to find the optima of  $ \tilde F(x,b) = \sum \hat f(x_i, b_ix_i)$ under the constraints $\|x\|_2^2=1$ and $b_i=\pm 1$ for all $i$.  It is easy to see that for any fixed $b=b_0$, $$\sup_x \tilde F(x,b_0) \leq \max\{\sup_x \tilde F(x,(1\dots 1)), \sup_x \tilde F(x,-(1\dots 1))\ .$$
Hence it suffices to find the optima of $\tilde F_{+1}(x)=\sum x_i^2\log x_i^2$, which are known from standard
information theory to be $0$ and $n\log n$.

\section{Proof of Lemma~\ref{lemma2appendixA}}\label{appendixproof}
\begin{proof}
Define the  function $g:\R\mapsto \R$ as 
$$ g(z) = - \frac 1{\sqrt n}\left ( \frac 1 {\sqrt n} + z\right ) \log\left |  \frac 1{\sqrt n}\left ( \frac 1 {\sqrt n} + z\right )\right |\ .$$
Let $g'(z)$ define the derivative of $g$ (where exists, namely for $z\neq -1/\sqrt{n}$).  Clearly,
$$ g'(z) = -\frac 1 {\sqrt n} \left ( \log \left |  \frac 1{\sqrt n}\left ( \frac 1 {\sqrt n} + z\right )\right | + 1\right )\ .$$
We split the sum in (\ref{lemma2appendixAeq}) to $i\in I^+ := \{i': \eps_{i'} \geq 0\}$ and $i\in I^- := [n]\setminus I^+$.
For the former case, note that $g$ is monotonically increasing in the range $[0,1/2]$.  Hence,
\begin{equation}\label{g1} \sum_{i\in I^+} g(\eps_i) \geq |I^+|\cdot g(0) =  \frac {|I^+|} n \log n \ .
\end{equation}
(We also used the fact that $|\eps_i|<1/2$ for all $i$.)
For $i\in I^-$, define $z_0<-\frac 1 {\sqrt n}$ to be the unique number such that
\begin{equation}\label{defC0}
\frac{g(0) - g(z_0)}{-z_0} = g'(z_0)\ .
\end{equation}
(This is the unique point to the left of $-\frac 1 {\sqrt n}$ such that the tangent line to $g$ at that point intersects the vertical line $z=0$ at $(0,g(0)) = \left(0,\frac 1 n \log n\right)$.)
It is trivial to show that 
\begin{equation}\label{kkk}
-1\leq z_0 \leq  -\frac{C_1}{\sqrt n} \log n
\end{equation}
 for some global $C_1>0$.
It is also trivial to show that for all $z<0$,
$$ g(z) \geq -g(z_0) + (z-z_0)g'(z_0) = g(0) + z\cdot g'(z_0) = \frac 1 n \log n  - \frac z {\sqrt n}\left ( \log \left |  \frac 1{\sqrt n}\left ( \frac 1 {\sqrt n} + z_0\right )\right | + 1\right )\ ,$$
namely, the graph of $g$ (in the left half plane) lies above the tangent at $z_0$.  This implies, also using (\ref{kkk}), that
$$ \sum_{i\in I^-} g(\eps_i) \geq \frac{|I^-|}{n}\log n + \sum_{i\in I^-} \frac{\eps_i}{\sqrt n} (C_2 \log n + C_3)\ ,$$
for some global $C_2>0$ and  $C_3$.
By setting $C_0$ appropriately and recalling that $\|\eps\|_1\leq \sqrt n \|\eps\|_2$, we get
\begin{equation}\label{g2}
\sum_{i\in I^-} g(\eps_i)  \geq \frac{|I^-|}{n}\log n  - \frac 1 4 \log n\ .
\end{equation}
Combining (\ref{g1}) and (\ref{g2}), we conclude the required.
\end{proof}

\end{document}